\documentclass[12pt]{article}

\usepackage[utf8]{inputenc}
\usepackage{graphicx} 
\usepackage{color} 
\usepackage{url}
\usepackage{amsmath}
\usepackage{amsthm}
\usepackage{amssymb}
\usepackage[margin=1in]{geometry}
\usepackage{appendix}
\numberwithin{equation}{section}

\newtheorem{theorem}{Theorem}[section]
\newtheorem*{theorem*}{Theorem}
\newtheorem{proposition}[theorem]{Proposition}
\newtheorem{lemma}[theorem]{Lemma}
\newtheorem{corollary}[theorem]{Corollary}

\newtheorem{remark}[theorem]{Remark}

\newtheorem{definition}[theorem]{Definition}
\newcommand{\RR}{\mathbb{R}}

\newcommand{\CC}{\mathbb{C}}
\newcommand{\trace}{\mathrm{trace}}

\newcommand{\B}{\mathcal{B}}
\newcommand{\M}{\mathcal{M}}

\newcommand{\N}{\mathcal{N}}
\newcommand{\X}{\mathcal{X}}
\newcommand{\Y}{\mathcal{Y}}
\renewcommand{\L}{\mathcal{L}}

\newcommand{\A}{\mathcal{A}}

\newcommand{\R}{\mathbb{R}}

\newcommand{\w}{\mathbf{w}}
\DeclareMathOperator{\Id}{Id}
\DeclareMathOperator{\GL}{GL}
\DeclareMathOperator{\Span}{span}
\DeclareMathOperator{\im}{im}
\DeclareMathOperator{\SO}{SO}
\DeclareMathOperator{\OO}{O}
\DeclareMathOperator{\Proj}{Proj}
\DeclareMathOperator{\Gr}{Gr}
\DeclareMathOperator{\Hom}{Hom}

\newcommand{\rev}[1]{{\color{black} #1}}

\usepackage{comment}

\title{Phase Retrieval with Semi-algebraic and ReLU Neural Network Priors}

\begin{document}


\author{Tamir Bendory\thanks{School of Electrical and Computer Engineering, Tel Aviv University} \and Nadav Dym\thanks{Faculty of Mathematics, Technion-Israel Institute of Technology} \and Dan Edidin\thanks{Department of Mathematics, University of Missouri} \and Arun Suresh\thanks{Department of Mathematics, University of Missouri}}

\maketitle

\begin{abstract}
The key ingredient to retrieving a signal from its Fourier magnitudes---namely, to solve the phase retrieval problem---is an effective prior on the sought signal. 
In this paper, we study the phase retrieval problem under the prior that the signal lies in a semi-algebraic set. This is a very general prior as  semi-algebraic sets include linear models, sparse models, and 
 ReLU neural network generative models. The latter is the main motivation of this paper, due to the remarkable success of deep generative models in a variety of imaging tasks, including phase retrieval. We prove that almost all signals in $\R^N$ can be determined from their Fourier magnitudes, up to a sign, if they lie in a (generic) semi-algebraic set of dimension $N/2$. The same is true for all signals if the semi-algebraic set is of dimension $N/4$.
We also generalize these results to the problem of
signal recovery from the second moment in multi-reference alignment models with multiplicity-free representations of compact groups.
This general result is then used to derive improved sample complexity bounds
for recovering band-limited functions on the sphere from their noisy copies, each acted upon by a random element of $\SO(3)$.	
\end{abstract}



\section{Introduction}

We consider the \emph{phase retrieval problem} of recovering a signal $x \in \R^N$ from the magnitudes of its discrete Fourier transform, namely, from its power spectrum. 
Phase retrieval is the main computational problem in a wide variety of applications, including X-ray crystallography~\cite{harrison1993phase}, optical imaging~\cite{shechtman2015phase}, ptychography~\cite{yeh2015experimental}, and astronomy~\cite{fienup1987phase}.
Since the power spectrum determines the signal only up to the set of $N$ phases of the Fourier transform, the set of signals with the same power spectrum forms an orbit under a product
of $N$ circles.
To narrow down the set of possible solutions, two main strategies are used: considering a prior on the signal or acquiring additional measurements. 
An example of the latter is ptychography, in which the power spectra of overlapping patches are recorded~\cite{marchesini2016alternating,bendory2017non,iwen2020phase}. 
A prominent example of the former is the computational problem in X-ray crystallography, where the signal is assumed to be sparse~\cite{elser2018benchmark}.

Our inquiry in this paper is motivated by the remarkable success of using ReLU networks as generative models for phase retrieval applications, e.g.,~\cite{sinha2017lensless,rivenson2018phase,metzler2018prdeep,hand2018phase,deng2020learning},  as well as for related imaging tasks~\cite{ongie2020deep,zhong2021cryodrgn}. Motivated by this success, we would like to address signals $x\in \RR^N$ which arise from a neural network generative model
\begin{equation} \label{eq:relu}
x=A_\ell\circ \sigma_{\ell-1} \circ A_{\ell-1}\circ \ldots \circ \sigma_1 \circ A_1(z),
\end{equation}
where $z$ resides in a low-dimensional space, the $A_i$'s are linear transformations and the $\sigma_i$'s are semi-algebraic activation functions, such as the rectification function (ReLU) that sets the negative entries of an input to zero. We would like to know whether the phaseless Fourier measurements of this signal determine the signal up to a global sign. This is part of an ongoing effort to unveil the mathematical foundations of the phase retrieval problem; 
see for example recent results on the geometry of phase retrieval~\cite{beinert2015ambiguities,edidin2019geometry,barnett2022geometry}, and specifically on X-ray crystallography~\cite{bendory2020toward,bendory2023finite,edidin2023generic,ghosh2022sparse,ranieri2013phase}. 
For recent surveys on the mathematics of phase retrieval, see~\cite{bendory2017fourier,grohs2020phase,bendory2022algebraic}.

In this paper, we address the phase retrieval problem under the general
assumption that we know that our $N$-dimensional signal $x$ resides in
a semi-algebraic set $\M$ of dimension $M<N $. This prior is very
general and includes \rev{realistic priors, such as} the collection of all signals in $\RR^N$, which
can be obtained as images of ReLU networks as in~\eqref{eq:relu}, as
well as linear and sparse models. \rev{The latter is essential to the mathematical framework of X-ray crystallography, facilitating the reconstruction of molecular three-dimensional structures~\cite{elser2018benchmark,edidin2023generic}.}
Our main results
(Theorems~\ref{thm.main} and Theorem~\ref{thm.rotation}) state that if $M <
N/2$, then it is possible to solve the phase retrieval problem on the generic linear
translate or generic rotation of {\em any} semi-algebraic set. As a
corollary (Corollary~\ref{cor.deepnetwork}), we prove that it is
possible to solve the phase retrieval problem for signals that are in the image of a
very general class of ReLU neural networks. 
\rev{Corollary~\ref{cor:sparse} states the result for sparse priors.}
Moreover, the mathematical tools we develop to study the phase retrieval problem can also be applied to
more general multi-reference alignment problems, \rev{which are relevant for the single-particle cryo-electron microscopy technology used to elucidate the structure of biological molecules,} as we discuss in
Section~\ref{sec.MRAmultone}. 
    
\subsection{Mathematical preliminaries and notation}
The statements of our results require some basic terminology
from real algebraic geometry. A semi-algebraic set $\M \subseteq \RR^N$ is a finite union of sets, which are defined by polynomial equality and inequality constraints. Any semi-algebraic set~$\M$ can be written as a finite union of smooth manifolds, and the dimension of~$\M$ is defined to be the maximal dimension of these manifolds. We say that a property holds for \emph{generic} $x$ in~$\M$, if the property holds for all $x$ outside a semi-algebraic subset $\N$ of $\M$ of strictly lower dimension. 
If $\M$ is algebraic, this is equivalent to saying that there is a non-empty Zariski dense open set $U \subset \M$ such that the property holds for $x \in U$. Additional background on algebraic geometry can be found in Appendix~\ref{app.algebraic} and books such as \cite{basu2006algorithms,bochnak1998real}.

Throughout this paper, we say that two signals $x,y\in \RR^N$ are
  {\em equivalent} if $x = \pm y$ and we denote this by writing
  $x\sim y$.

\subsection{Main results} \label{sec:main_results}
\rev{The set of signals with the same power spectrum forms an orbit under the group $\{\pm 1\} \times (S^1)^{(N-1)/2}$ ($N$ odd)
or $\{\pm 1\}^2 \times (S^1)^{(N-2)/2}$ ($N$ even)~\cite[Section 5]{bendory2020toward} so in order to recover a signal from its power spectrum we must impose a prior
condition on the signal.}
Our first main result states that given any semi-algebraic set $\M\subset \R^N$ 
of sufficiently small dimension, the power spectrum
separates signals that lie in the translated set $A \cdot \M =\{Ax|\, x \in \M\}$, when  $A$ is a generic linear transformation on $\R^N$.

\begin{theorem} \label{thm.main}
  Let $\mathcal{M}$ be a semi-algebraic subset of $\RR^N$ of dimension $M$. 
 Then, for a generic invertible matrix $A \in \GL(N)$, the following hold:
\begin{enumerate}
\item If $N\geq 2M$,  
  then a generic vector $x \in A\cdot \M$ is determined, up to a sign, by its power spectrum.
    \item If $N \geq 4M$, then every vector $x \in A \cdot \M$ is determined, up to a sign, by its power spectrum.  
\end{enumerate}
\end{theorem}

\rev{The power spectrum of a vector in $\R^N$ consists of $\sim N/2$ measurements and Theorem~\ref{thm.main}
implies that the power spectrum is injective on a `generic' semi-algebraic set of dimension at most $N/4$, which is half the number of measurements.
Essentially, the same tradeoff between the number of measurements and dimension occurs in frame phase retrieval~\cite{balan2006signal}.
Given an $N$ element frame in $\R^N$, the real frame phase retrieval problem is the problem of recovering a signal in $\R^M$ from $N$ phaseless frame measurements. This is equivalent to recovering all points (up to a global sign) in an $M$-dimensional linear subspace of $\R^N$ (the range of the analysis operator of the frame) from the absolute values of its coordinates. 
The result of~\cite{balan2006signal} states that a generic frame of size $N \geq 2M -1$ does phase retrieval.  Phrased in the language of this paper, this means that if $\M$ is a generic linear subspace of $\R^N$ of dimension $M$ with 
$M \leq (N+1)/2$, then the 
map $\R^N \to \R^N_{\geq 0}, (x_1, \ldots , x_N) \mapsto (|x_1|, \ldots , |x_N|)$ is injective restricted to $\M$.}

Our second result shows that if $\M$ is a semi-algebraic set of
sufficiently small dimension, then a generic rotation
of $\M$ separates signals, up to a sign, from their power spectra. 
This result is arguably more natural, since if $A \in \SO(N)$ is
a rotation matrix, then $A \cdot \M$ is isometrically equivalent to $\M$.

\begin{theorem} \label{thm.rotation}
  Let $\mathcal{M}$ be a semi-algebraic subset of $\RR^N$ of dimension $M$. Then,  for a generic rotation $A \in \SO(N)$, the following hold: 
\begin{enumerate}
\item If $N\geq 2M+2$,  
  then a generic vector $x \in A\cdot \M$  is determined, up to a sign, by its power spectrum.
\item If  $N\geq 4M+2$, then every vector $x \in A \cdot \M$ is determined,  up to a sign, by its power spectrum.
\end{enumerate}
\end{theorem}

We note that Theorem~\ref{thm.main} and Theorem~\ref{thm.rotation} are independent results and require separate proofs. However, the proof
  of Theorem \ref{thm.rotation} is significantly more intricate than
  the proof of Theorem \ref{thm.main}.

\subsection{Applications}

As mentioned earlier, 
there are two natural situations where we can apply our theoretical results. We first consider signals generated by neural networks.

The following corollary states that a signal can be recovered from its power spectrum, up to a sign, if it was generated by a neural network whose last layer's width is two (respectively four) times larger than the dimension of the network's image. 
The activations can be taken to be ReLU functions or other semi-algebraic activation functions, such as leaky ReLU, PReLU, or Hardtanh~\cite{pytorch}. The corollary is proven in Section~\ref{sub.NN}.

\begin{corollary}[Deep neural networks] \label{cor.deepnetwork}
  Let $\Phi: \RR^K \to \RR^{N_{\ell-1}}$ be a neural network of the form \begin{equation*}
\Phi =\sigma_{\ell-1}\circ A_{\ell-1}\circ \ldots \circ \sigma_1 \circ A_1,\end{equation*}
 where the activation functions $\sigma_i:\RR \to \RR$ are semi-algebraic functions. Denote the dimension of the image  $\Phi(\RR^K)$ by $M$ and let $A_\ell$ be a  generic $N \times N_{\ell-1} $ matrix.  Then:
\begin{enumerate}
\item if 
$N \geq 2M$, then a generic vector in $\M = \im A_\ell \circ  \Phi$ can be recovered, up to a sign,  from its power spectrum.
\item if 
$N \geq 4M$, then any vector in $\M = \im A_\ell \circ  \Phi$ can be recovered, up to a sign,  from its power spectrum.
\end{enumerate}
\end{corollary}  
Note that $M$ will always be smaller than or equal to the minimal dimension of all layers in $\Phi$. Thus, the corollary implies that phase retrieval has a unique solution for a generic neural generative model as in~\eqref{eq:relu} whose last layer is at least two (respectively, four) times larger than the minimal layer.
\rev{Importantly, the requirement that $A_\ell$ be a generic matrix is not a significant limitation. In practice, this condition can be satisfied by introducing a small perturbation to the final layer of a trained neural network.}

Motivated by X-ray crystallography, we also consider the prior that
the signal is sparse. In this case, $\M$ is the union of the
$\binom{N}{M}$ linear subspaces spanned by the vectors $\{e_i\}_{i \in
  S}$, where $S$ runs through all subsets of $[1,N]$ of size
$M$. If $A \in \GL(N)$ is generic, then $A\cdot \M$ is the set of vectors that are sparse
with respect to a generic basis for $\RR^N$. 
In this case, our results 
give slightly weaker bounds  
than those obtained in~\cite{edidin2023generic} using projective
algebraic geometry.
However, if~$A$ is a rotation matrix, then
$A\cdot \M$ is the set of vectors which are $M$-sparse when expanded with respect to
the orthonormal basis given by the rows of $A$, $\{Ae_i\}_{i=1, \ldots N}$; this case has not been studied before. Thus, applying Theorem \ref{thm.rotation}, we obtain the following new result.
\rev{
\begin{corollary}[Sparsity]\label{cor:sparse}
Let $\mathcal{B}$ be a generic orthonormal basis and let $\M$ be the set of signals which are $M$-sparse when expanded in the basis
$\mathcal{B}$.
  \begin{enumerate}
    \item If $N \geq  2M+2$, then a generic signal in $\M$ can be recovered, up to a sign, from its power spectrum.
      \item If $N \geq 4M+2$, then every
       signal in $\M$ can be recovered, up to a sign, from its power spectrum.
  \end{enumerate}
\end{corollary}
}

\rev{The problem of recovering signals which are sparse with respect to the standard orthonormal basis has a long history in the literature, going back to the work of Patterson on binary signals~\cite{patterson1944ambiguities}. More recently, the papers~\cite{ranieri2013phase,jaganathan2012recovery, elser2018benchmark, bendory2023finite} have studied phase retrieval with sparse signals where the non-zero entries can vary in a finite alphabet or continuously.}
In~\cite{bendory2020toward}, it is conjectured (and verified for a limited number of cases)  that for the standard orthonormal
basis, a generic sparse vector whose support set satisfies
$|S -S| > |S|$ can be recovered from its power spectrum up to a dihedral (circular shifts
and reflection) equivalence. A necessary, but not sufficient,
condition for this to happen is  $|S| < N/2$. Moreover, the experiments in \cite{bendory2020toward} indicate that, empirically, 
a subset $S \subset [1,N]$ with $|S| < N/2$ 
is  likely to satisfy
$|S -S| > |S|$ but for any $N$ there are always subsets $S$ with $|S| < N/2$
that do not satisfy this condition. 
\rev{Notably, Corollary~\ref{cor:sparse} does not apply in this case, as the genericity condition does not necessarily encompass the standard basis. Yet,} if we set $M = |S|$, then the bounds obtained here for generic orthonormal bases are comparable
to the conjectural bounds of \cite{bendory2020toward}.
Moreover,  by applying a generic rotation, we eliminate the dihedral ambiguity
of \cite{bendory2020toward}.

\subsection{Organization of the paper} 
Section \ref{sec:common_strategy}
  describes the overall strategy and introduces the techniques
  necessary for proving the main results of the paper.
  Theorem \ref{thm.main} is proved in Section~\ref{sec:proof_thm_main} and
  Theorem \ref{thm.rotation} is proved in Section~\ref{sec:proof_thm_rotation}.
  In Section~\ref{sec.MRAmultone} we show how the methods used
  to prove Theorems~\ref{thm.main} and \ref{thm.rotation} can be used
  to extend our results from phase retrieval to the multi-reference
  alignment problem of recovering a signal in a multiplicity-free
  representation of a compact group from its second moment.
Finally, in Section~\ref{sec:future_work} we outline some future research directions.

\section{Common strategy for both proofs} \label{sec:common_strategy}
\subsection{The power spectrum and the second moment} \label{sec.2ndmoment}

In this paper, we are interested in the problem of recovering a
signal in $\R^N$ from its power spectrum or, equivalently,  from its
periodic auto-correlation.  To do this, we will follow a strategy
developed in ~\cite{edidin2023generic} and use the fact that the power
spectrum is equivalent to the second moment for the action of the dihedral or the
cyclic group on $\RR^N$ (the second moments are the same~\cite{bendory2022sample}).

Following ~\cite{edidin2023generic}, if we expand a signal  $x
\in \RR^N$ with respect to a real Fourier basis $v_1,
\ldots , v_{N}$, and we write $x = \sum_{n=1}^{N} x[n] v_n$,  
then the second moment is given by the formula\footnote{Note that the indexing we use here differs from that used in \cite{edidin2023generic}.
There, the indices go from $0$ to $N-1$ and the components are given as $x[n]^2 + x[N-n]^2$.
The latter indexing is more natural from a Fourier theory perspective. Our indexing was chosen
to simplify the notation in the proofs of our results.}:

\begin{equation}
  b_x = (x[1]^2, x[2]^2, x[3]^2 + x[4]^2, \ldots , x[N-1]^2 + x[N]^2),
\end{equation}
if $N$ is even, and
\begin{equation}
 b_x = (x[1]^2, x[2]^2 + x[3]^2, \ldots , x[N-1]^2 + x[N]^2),   
\end{equation}
if $N$ is odd.
Now, consider a linear map $\RR^N \stackrel{A} \to \RR^N$, given by
\begin{equation}
 x \mapsto \left( \langle x,w_1 \rangle, \ldots , \langle x, w_{N} \rangle\right),
\end{equation}
where $x$ is represented by a column vector and $w_1, \ldots, w_{N}$ are the rows
of the matrix representing $A$ with respect to our chosen basis for $\RR^N$. In particular, if $x \in \R^N$ then, using the notation of \cite{dym2022low},
\begin{equation*}
b_{Ax} = P(x;A),\end{equation*}

where
\begin{equation*}
  P(x;A) = \left( \langle x, w_1 \rangle^2, \langle x,w_2 \rangle^2,
  \langle x,w_3 \rangle^2 + \langle x, w_4 \rangle^2, \ldots, \langle x,w_{N-1}
  \rangle^2 + \langle x, w_N \rangle^2 \right),
\end{equation*}  if $N$ is even
and
\begin{equation*}
  P(x;A) = \left( \langle x, w_1 \rangle^2, 
  \langle x,w_2 \rangle^2 + \langle x, w_3 \rangle^2, \ldots, \langle x,w_{N-1}
  \rangle^2 + \langle x, w_N \rangle^2 \right),
\end{equation*}
if $N$ is odd.
This implies that if $A$ is an arbitrary $N \times N$ matrix, 
then (with respect to the real Fourier basis) the 
phaseless Fourier measurements of $Ax$ are {\em separable}~\cite{dym2022low}
with respect to the rows of $A$.

\subsection{Reduction to fiber dimension counting}
Our main results, Theorem~\ref{thm.main} and Theorem~\ref{thm.rotation}, are similar in structure, but differ in that the first considers matrices $A$ coming from $\GL(N)$, while the latter considers matrices from $\SO(N)$. As a preliminary step for both proofs, we prove a lemma that holds for general semi-algebraic subsets $\A \subseteq \RR^{N\times N} $. This lemma will reduce both theorems
to the task of uniformly bounding  the dimension of the sets
\begin{equation*}
\A(x,y)=\{A\in \A\,|\, P(x;A)=P(y;A) \},
\end{equation*}
over all $x,y \in \M$ which are not equivalent.
\begin{lemma} \label{lem.AA}
  Let $\mathcal{M}$ be a semi-algebraic subset of $\RR^N$ of dimension $M$, and let $\A \subseteq \RR^{N\times N}$ be a semi-algebraic set such that
\begin{equation*}
\dim(\A(x,y))\leq \dim(\A)-K, \quad  \forall x,y\in \M, \text{ with } x\not \sim y.\end{equation*}

  Then, for a generic invertible matrix $A \in \A$, the following hold:
\begin{enumerate}
\item If $K > M$,
  then a generic vector $x \in A\cdot \M$  is determined, up to a sign, from its power spectrum.
\item If $K > 2M$, then every vector $x \in A \cdot \M$ is determined, up to
 a sign, from its power spectrum.
\end{enumerate}
\end{lemma}

\begin{remark}
  Since we are typically interested in proving a result for a general
  semi-algebraic subset $\M \subset \R^N$ without any special knowledge
  of $\M$, when we apply Lemma~\ref{lem.AA} we will bound $\dim(\A(x,y))$ over all pairs $x \not \sim y$ in $\R^N$,
  rather than over pairs in $\M$. \rev{We also note that using generic linear translates is indeed not enough to eliminate the sign ambiguity. The reason is that if $x,-x$ are in $\M$,  then $Ax, -Ax \in A \cdot \M$ for \emph{every} linear transformation $A$.}
\end{remark}

\begin{proof}
The proof of this lemma follows the guidelines of the proofs of Theorem 1.7 and Theorem 3.3 in~\cite{dym2022low}.
The proof is based on studying
 the following semi-algebraic incidence set:
 \begin{equation}
  \B=\B(\M,\A)=\{(x,y,A)\in \M \times \M \times \A\,|\,  x\not \sim y \text{ but } P(x;A)=P(y;A)  \} \subset \M \times \M  \times \A.
\end{equation}
Let \rev{$\pi: \B \to \M \times \M$} be the projection
\begin{equation*}
\pi(x,y,A)=(x,y) .\end{equation*}

According to ~\cite[Lemma 1.8]{dym2022low}, we can bound the dimension of the incidence $\B$ by  
\begin{equation}\label{eq:bound}
\begin{split}
\dim(\B)&\leq\dim(\pi(\B))+\max_{(x,y)\in \M \times \M, x\not \sim y} \dim \left(\pi^{-1}(x,y) \right)\\ &\leq 2\dim(\M)+ \max_{(x,y)\in \M \times \M, x\not \sim y} \dim \A(x,y)\\
&\leq 2M+ \dim(\A)-K.
\end{split}	
\end{equation}

Now, if $K>2M$ then we obtain that $\dim(\B)<\dim(\A)$. If
  $\phi$ is the projection $\phi(x,y,A)=A $ then $\dim(\phi(\B))<\dim(\A)$.
  Thus, a generic $A\in \A$ is not in $\phi(\B)$, and for all such $A$ we know that if $P(x;A)=P(y;A) $ for $x,y\in \M$ then $x\sim y$. This concludes the proof of the second part of the lemma.

For the first part of the lemma, let us assume that $K>M $ so that $\dim(\B)<M+\dim(\A)$. We consider two cases: if $\dim(\phi(\B))<\dim(\A) $ then we have,
as before, that a generic matrix $A \in \A$ is not in $\phi(B)$ and that for all such $A$ we know that if $P(x;A)=P(y;A) $ $x, y \in \M$ then $x \sim y$.

In the second case, we assume that $\dim(\phi(\B))=\dim(\A)$
and we will prove that for a generic matrix $A \in \A$, the fiber
$\phi^{-1}(A) \cap \B$ has dimension strictly less than $M$. Granted
this fact, the proof of the first part proceeds as follows:
Since $\dim (\phi^{-1}(A) \cap \B) < M$, the projection of the fiber onto the $x$ coordinate, which is the set 
\begin{align}
	\N_A&=\{x\in \M| \exists y \in \M \text{ such that } (x,y,A)\in \B\}\\
	&=\{x\in \M| \exists y \in \M \text{ such that }x \not \sim y \text{ but }P(x;A)=P(y;A)\}, \nonumber
\end{align}
also has a dimension strictly less than $M$.
Thus, generic $x$ will not be in $\N_A$, and for such $x$, if there exists some $y\in \M$ such that $P(x;A)=P(y;A) $, then $x \sim y $.

We now prove that the generic fiber of $\phi \colon \B \to \A$ has dimension
less than $M$. 
By \cite[Proposition 2.9.10]{bochnak1998real} the semi-algebraic set
  $\A$
is a disjoint union of Nash submanifolds\footnote{\rev{A semi-algebraic subset of $A \subset \R^N$ is called a Nash submanifold of dimension $d$ if,  for every point $x \in A$, there is an open neighborhood $x \in U \subset \R^N$ and a $C^\infty$ differentiable
function $\varphi \colon U \to \R^d$ such that the restriction of $\varphi$ to $U \cap A$ is a diffeomorphism~\cite[Definition 2.9.9]{bochnak1998real}}.} $\A_1,\ldots,\A_s$. Clearly,
it suffices to prove that  $\dim (\phi^{-1}(A) \cap \B) < M$
for a generic matrix $A$ in any Nash submanifold $\A_\ell$ where $\dim
\A_{\ell} = \dim \A$.
With this observation, we can reduce to the case that $\A$ is a
manifold.

Now we decompose $\B$ as a disjoint union of Nash submanifolds
$\bigcup_{\ell}^L U_\ell$ and 
let $\L\subseteq \{1,\ldots,L\}$
denote the subset of indices $\ell$ for which $\phi(U_\ell)$ attains its maximal value $\dim(\A)$.
For every fixed $\ell\in \L$,
the semi-algebraic version of Sard's
Theorem~\cite[Theorem 9.6.2.]{bochnak1998real}
implies that the generic matrix $A \in \phi(U_\ell)$
is a regular value (i.e., not the image of a critical point) of the restriction of $\phi$ to $U_\ell$.
By the
pre-image theorem \cite[Theorem 9.9]{tu2011manifolds},
every regular value $A$ is
either not in the image of $\phi|_{U_\ell}$, or we have the equality
\begin{equation}
 \dim(U_\ell)=\dim \phi(U_\ell)+\dim (\phi^{-1}(A)\cap U_\ell)=\dim(\A) +
\dim (\phi^{-1}(A) \cap U_\ell).    
\end{equation}
So,
\begin{align*}\dim \left( \phi^{-1}(A) \cap U_\ell\right)=\dim(U_\ell)-\dim(\A)
\leq \dim(\B)-\dim(\A)<M.
\end{align*}
The semi-algebraic set, which is the union of
the images of the critical points of \rev{$\{\phi|_{U_\ell}\}_{\ell \in \L}$} and
$ \cup_{\ell \not \in \L}\phi(U_\ell)$ has dimension
strictly less than $\dim \A$. If $A \in \A$ lies in the complement
of this set, then we have the fiber
\begin{equation}
 \phi^{-1}(A)\cap \B=\cup_{\ell\in \L}\left( \phi^{-1}(A)\cap U_\ell \right),   
\end{equation}
has a dimension strictly smaller than $M$.
\end{proof}

\section{Proof of Theorem \ref{thm.main}} \label{sec:proof_thm_main}
We now proceed to prove Theorem~\ref{thm.main}, which deals with the case
when $\A=\GL(N) $. Since a generic matrix in $\RR^{N\times N}$ is invertible, we can equivalently prove Theorem~\ref{thm.main} where we replace $\GL(N) $ with all of $\RR^{N\times N}$. According to Lemma~\ref{lem.AA}, to do this, it is sufficient to prove the following lemma:

\begin{lemma}\label{lem.RNN}
Let $N$ be a natural number. For every $x,y\in \RR^N$ such that $x \not \sim y$, consider the semi-algebraic set 
\begin{equation*}
\RR^{N\times N}(x,y)=\{A\in \RR^{N \times N} \, | \, P(x;A)= P(y;A) \} .\end{equation*}

Then, $\dim(\RR^{N\times N}(x,y)) \leq N^2-K(N)$, where $K(N)=N/2+1$ (if $N$ is even) or
$K(N)=(N+1)/2$
(if $N$ is odd).
\end{lemma}
Indeed, using the first part of Lemma~\ref{lem.AA} we see that if the signals lie in a semi-algebraic set of dimension
$M$, then we require that $K=K(N)>M $. An elementary calculation shows that (for both $N$ even and $N$ odd) this is equivalent to the condition $N \geq 2M$. Similarly,  the second part of Lemma~\ref{lem.AA} requires that $K=K(N)>2M$, which is equivalent (for both $N$ even and $N$ odd)  to the condition $N \geq 4M$. Thus, to prove Theorem~\ref{thm.main} it remains only to prove  Lemma~\ref{lem.RNN}. 

\begin{proof}[Proof of Lemma~\ref{lem.RNN}]
Let us first consider the case when $N$ is even.
Writing $A=(w_1,\ldots,w_{N}) $ as before, we note that $P(x;A)=P(y;A) $ if and only if $w_1$ an $w_2$ are in the set
\begin{equation*}
W^1=W^1(x,y)=\{w\in \RR^N\, | \, p_1(x,w)=p_1(y,w) \} ,\end{equation*}
where $p_1(x,w) = \langle x,w \rangle^2$,
and all pairs $(w_{2k+1},w_{2k+2}) $ for $k=1,\ldots,N/2-1$ are in the set 
\begin{equation*}
W^2=W^2(x,y)=\{(w,w')\in \RR^N \times \RR^N\, | \,
p_2(x,w,w')=p_2(y,w,w') \},\end{equation*}
 where $p_2(x,w,w') = \langle x, w \rangle^2  + \langle x, w' \rangle^2$.
In other words, $\RR^{N\times
  N}(x,y)$ is a Cartesian product of two copies of $W^1$ and $N/2-1 $
copies of $W^2$.  We note that $w$ is in $W^1=W^1(x,y)$ if and only
if $w$ is orthogonal to $x-y $ or $x+y$. For all $x,y$ which are not
equivalent, both $x-y$ and $x+y$ are non-zero vectors and so
$W^1(x,y)$ is a union of two hyperplanes and is of dimension
$N-1$. Similarly, we can show that $W^2(x,y)$ is a proper algebraic subset of $\R^{2N}$
by observing that whenever $x
\not \sim y$, we can take $w$ to be any vector not in
$W^1(x,y)$ and taking
$w'=0$. With these choices, $p_2(x,w,w')\neq p_2(y,w,w') $ so the 
polynomial equation defining $W^2(x,y)$ is not automatically satisfied for all $(w,w') \in \R^N \times \R^N$. 
Therefore, $\dim(W^2(x,y))\leq 2N-1 $. Since the set $W^1(x,y)$ and
$W^2(x,y)$ both have codimension of at least one in $\R^N$ and $\R^{2N}$, respectively,
we conclude that 
\begin{equation*}
\dim \RR^{N\times N}(x,y)=2\dim(W^1(x,y))+(N/2-1)\dim(W^2(x,y))\leq N^2-(N/2+1).\end{equation*}

When $N$ is odd, $\RR^{N\times N}(x,y) $ is a Cartesian product of $(N+1)/2 $  sets: a single copy of $W^1$ and $(N-1)/2 $ copies of $W^2$, so that 
\begin{equation*}
\dim \RR^{N\times N}(x,y)\leq N^2-(N+1)/2.\end{equation*}

\end{proof}

\subsection{Implications for neural generative models: Proof of Corollary \ref{cor.deepnetwork}}\label{sub.NN}
Let $\Phi=\sigma_{\ell-1}\circ A_{\ell-1}\circ \ldots \circ \sigma_1 \circ A_1$ be our given neural network. For any matrix $A_\ell\in \R^{N \times N_{\ell-1}}$ set
 $\Phi_{A_\ell} = A_\ell \circ \Phi$. We wish to show
 that if $N \geq 4M$ (respectively, $N \geq 2M$) then for a
 generic choice of $A_\ell$  any 
 (respectively, a generic) vector in $\im \Phi_{A_\ell}$ can be recovered, up to a sign,  from its power spectrum. In either case (depending on the dimension
 bounds) we say that $\Phi_{A_\ell}$
 {\em satisfies the phase retrieval property}.

 To that end, let
 $\R^{N \times N_{\ell-1}}(\Phi)$ be the subset of $\R^{N \times N_{\ell-1}}$
 such that 
 $\Phi_{A_{\ell}}$ satisfies the phase retrieval property.
 We wish to
 show that the complement of $\R^{N \times N_{\ell-1}}(\Phi)$ in
 $\R^{N \times N_{\ell-1}}$ has dimension smaller
 than $\dim \R^{N \times N_{\ell -1}} =  NN_{\ell -1}$, or equivalently, that
 $\R^{N \times N_{\ell -1}}(\Phi)$
 contains a non-empty Zariski open set. Let $F(N,N_{\ell-1})$
 be the Zariski open subset of $\R^{N \times N_{\ell-1}}$, parametrizing matrices
 of full rank. Clearly, it suffices to prove that 
 \begin{equation*}
\R^{N \times N_{\ell -1}}(\Phi) \cap
 F(N, N_{\ell -1})\end{equation*}
 contains a non-empty Zariski open set; i.e., we need only consider matrices $A_\ell$ of full-rank.
 We have two cases to consider: $N_{\ell-1} \leq N$, and $N_{\ell-1} > N$.

If $N_{\ell-1} \leq N$, then any full-rank matrix $A_\ell \in
F(N, N_{\ell -1})$ can be obtained as $A A_\ell^0$ where $A \in
\GL(N)$ and $A_\ell^0$ is a fixed full-rank matrix. Let 
$\M_0 = \im
\Phi_{A_\ell^0}$.  By Theorem~\ref{thm.main} we know that there is a
Zariski open set $U \subset \GL(N)$ such that (with the appropriate
dimension bounds) the power spectrum separates (generic) signals up to
sign on $A \cdot \M_0$ when $A \in U$. Now, $A \cdot \M_0 = \im \Phi_{AA_\ell^0}$, so 
$AA_\ell^0 \in \R^{N \times N_{\ell -1}}(\Phi)$ for every $A \in U$. The set $\{A A_{\ell-1}^0|
A\in U\}$ forms a Zariski open set in $F(N, N_{\ell -1})$ and thus
in $\R^{N \times N_{\ell -1}}$ so the corollary follows in this
case.

Now, suppose that $N_{\ell-1} >N$. Here we will use~\eqref{eq:conservation2} to show
the complement of $\R^{N \times N_{\ell -1}}(\Phi)$
in $F(N,N_{\ell -1})$ has dimension strictly smaller than $N_{\ell-1} N$.
Consider the action of $\GL(N)$ on $F(N,N_{\ell -1})$ given by
$A \cdot A_\ell = AA_\ell$. The orbit space for this action
is the Grassmann manifold $\Gr(N, N_{\ell-1})$ parametrizing $N$-dimensional
linear subspaces in $\R^{N_{\ell-1}}$.
For any fixed $N$-dimensional subspace, $L \subset \R^{N_{\ell-1}}$ the fiber
of the quotient map $F(N,N_{\ell -1}) \to \Gr(N, N_{\ell-1})$ over
the point corresponding to $L$ is the set of $N \times N_{\ell -1}$ matrices
whose range is $L$. If we fix one such matrix $A_\ell^{0}$, then
every other matrix in the fiber is row equivalent to $A_\ell^{0}$
and thus has the form $A A_\ell^0$ where
$A \in \GL(N)$. By Theorem \ref{thm.main}, $\Phi_{AA_\ell}$ satisfies phase retrieval for a generic choice
of $A \in \GL(N)$.
Hence, the set of $A_\ell$ in the fiber such
that $\Phi_{A_\ell}$ does not satisfy phase retrieval
has a dimension strictly smaller
than $N^2 = \dim \GL(N)$. Now, applying~\eqref{eq:conservation2} we conclude
that the set of $A_\ell \in 
F(N, N_{\ell -1})$ such that
$\Phi_{A_\ell}$ does not satisfy phase retrieval also has dimension
strictly smaller than $N N_{\ell -1}$ and the corollary is proved.

\section{Proof of Theorem~\ref{thm.rotation}} \label{sec:proof_thm_rotation}

Our goal in this section is to prove Theorem~\ref{thm.rotation}, which deals with orthogonal matrices $A \in \SO(N) $. Due to Lemma~\ref{lem.AA}, this is reduced to the task of bounding, for all
$x\not\sim y \in \RR^N$, the dimension of the set 
\begin{equation*}
\SO(x,y)  = \{A \in \SO(N)\,|\, P(x;A) = P(y;A)\}.\end{equation*}
 
\begin{proposition}\label{prop.crux}
  For any $x, y \in \R^N$ with $x \not \sim y$, 
  \begin{equation*}
\dim \SO(x,y) \leq \dim \SO(N) -N/2,\end{equation*}
 if $N$ is even
  and,
\begin{equation*}
 \dim \SO(x,y) \leq \dim \SO(N) - (N-1)/2,\end{equation*}
 if $N$ is odd.
\end{proposition}

The proof of proposition~\ref{prop.crux} is quite delicate, and as such, we will devote the next section to its proof. For notational simplicity,
we assume that $N$ is odd throughout the section. 

\subsection{Proof of Proposition \ref{prop.crux}}
We start by noting that an element $A \in \SO(N)$ can be viewed as an $N$-tuple of unit vectors $(w_1,\dots, w_N)$ characterized by the property that $w_k \in \Span(w_1,\dots, w_{k-1})^{\perp}$. Viewed this way, we realize $\SO(N)$ as
the last step in a tower of sphere bundles
\begin{equation*}
\SO(N)=F_N \to F_{N-1} \to \dots \to F_1 = S^{N-1},\end{equation*}
where $F_k$ consists of all unit vectors $(w_1,\dots, w_k)$ such that $w_k \in \Span(w_1,\dots, w_{k-1})^{\perp}$ and $(w_1,\ldots,w_{k-1})\in F_{k-1}$.
It is worth remarking that this process determines $w_N$ up to a sign. However, the requirement $\det(A) = 1$ determines this sign,  so, the fibers
of $F_N\to F_{N-1}$ consist of a single unit vector.

Let \begin{equation*}
F_1(x,y) = \{w \in F_1 = S^{N-1}\,|\,p_1(x,w) = p_1(y,w)\},\end{equation*}
 and
for $k \geq 1$, define \begin{equation*}
F_{2k+1}(x,y) = \{(\mathbf{w}, w_{2k}, w_{2k+1}) \in F_{2k+1}\,|\, \mathbf{w} \in
F_{2k-1}(x,y), \, p_2(x, w_{2k}, w_{2k+1}) = p_2(y, w_{2k}, w_{2k+1})\}.\end{equation*}

\begin{lemma}\label{lem.dimFxy}
 We have the following bound on the dimension of $F_{2k+1}(x,y)$:
\begin{equation*}
\dim F_{2k+1}(x,y)
\leq \dim F_{2k+1}-k-1,\end{equation*}

for $k=0,\dots, \dfrac{N-3}{2}$.
\end{lemma}

Lemma~\ref{lem.dimFxy} is a special case of Lemma~\ref{lem.Fxy_multfree}
which will be given in Section~\ref{sec.cruxproof}.
However, given Lemma~\ref{lem.dimFxy}, Proposition~\ref{prop.crux} follows immediately.

\begin{proof}[Proof of Proposition \ref{prop.crux}] Letting
  $k=\frac{N-3}{2}$ in the statement of Lemma \ref{lem.dimFxy} we see that $\dim F_{N-2}(x,y) \leq \dim F_{N-2} -\frac{N-1}{2}$.
Since the fibers of $F_N = \SO(N) \to F_{N-2}$ have dimension $1$,
\begin{equation*}
\dim \SO(x,y) \leq 1+\dim(F_{N-2}(x,y)),\end{equation*}
 and  
\begin{equation*}
\dim \SO(N)= 1+\dim(F_{N-2}).\end{equation*}

We conclude that $\dim \SO(x,y) \leq \dim \SO(N) -\frac{N-1}{2}$.
\end{proof}

\section{Multi-reference alignment with  irreducibles with multiplicity one} \label{sec.MRAmultone}

\subsection{Introduction}

Multi-reference alignment (MRA) is the problem of recovering a signal from its multiple noisy copies, each acted upon by a random group element. 
Let $G$ be a compact group acting on a vector space $V$. 
The MRA model reads
\begin{equation} \label{eq:mra1}
    y= g\cdot x + \varepsilon,
\end{equation}
{where $g\in G$, $x\in V$ and $g \cdot x$ denotes the translation of
$x$ by $g$ under group action. Finally, $\varepsilon\overset{i.i.d.}{\sim}\mathcal{N}(0,\sigma^2 I)$ models the noise.
The goal is to estimate the $G$-orbit of $x$ from $n$ observations $y_1,\ldots,y_n$:
\begin{equation} \label{eq:mra}
    y_i = g_i\cdot x + \varepsilon_i, \quad i=1,\ldots,n.
\end{equation}
where we assume the group elements are drawn from a uniform (Haar) distribution over~$G$. 
 
A wide variety of MRA models have been studied in recent years (see for example~\cite{perry2019sample,bendory2017bispectrum,bendory2022super,bandeira2023estimation}), mainly motivated by the single-particle cryo-electron microscopy (cryo-EM) technology, as discussed in Section~\ref{sec:cryoEM}.
In~\cite{abbe2018estimation,perry2019sample}, an intimate connection between the method of moments and the sample complexity of the MRA problem was uncovered. 
In particular, it was shown that in the high-noise regime as $n,\sigma\to\infty$, the minimal number of observations required to estimate the G-orbit of $x$ uniquely is $n/\sigma^{2d}\to\infty$  
, where $d$ is the lowest order moment of the observations that determines the $G$-orbit\footnote{This is true only when the dimension of the signal is finite~\cite{romanov2021multi}.}.   While it was shown that in many cases the third moment is required to determine a generic signal (e.g.,~\cite{bendory2017bispectrum,bandeira2023estimation,edidin2023orbit}), and thus the sample complexity is $n/\sigma^{6}\to\infty$, recent papers showed that under prior information (specifically, sparsity), the signal can be recovered from the second moment, implying an improved sample complexity of   $n/\sigma^{4}\to\infty$~\cite{bendory2022sample,bendory2023autocorrelation}.

In what follows, we show that the techniques of this paper can be readily extended to study the second moment of MRA models for multiplicity-free representations of compact groups.
Recall that the population second moment is given by 
\begin{equation}\label{eq:pop_second_moment}
 \mathbb{E}[yy^T] = \int_{g\in G} (g\cdot x) (g\cdot x)^Tdg +\sigma^2I.  
\end{equation}
Assuming the noise level $\sigma^2$ is known, the population second moment can be estimated from the observations by averaging over the empirical second moment as long as $n/\sigma^4\to\infty.$  Our goal will be to show that, with the prior that our signals lie in a low-dimensional semi-algebraic set, the population second moment determines a generic signal uniquely.

When $x\in\RR^N$ and $G$ is the group of circular shifts (or the dihedral group), the second moment is the periodic auto-correlation, which is equivalent to the power spectrum (via the Fourier transform). Thus, the phase retrieval problem can be thought of as a special case of the MRA model.

\subsection{Main results of MRA for multiplicity-free representations}
\begin{definition}
  A (real) representation $V$ of a compact 
 group is {\em multiplicity free}
  if $V$ decomposes as a sum of distinct (non-isomorphic) irreducible
  representations
  \begin{equation*}
V = V_1 \oplus \ldots \oplus V_R.\end{equation*}

\end{definition}

For multiplicity-free representations, we obtain the following extensions of
Theorem~\ref{thm.main} and Theorem~\ref{thm.rotation}.

\begin{theorem} \label{thm.main_multfree}
  Let $V$ be an $N$-dimensional real representation of a compact group $G$
  which decomposes as a sum of $R$ distinct irreducible representations
   $V  = V_1 \oplus \ldots \oplus V_R.$
  Let $\mathcal{M}$ be a semi-algebraic subset of $\RR^N$ of dimension $M$. Then, for a generic invertible matrix $A \in \GL(N)$, the following hold:
\begin{enumerate}
\item If $R > M$, 
  then a generic vector $x \in A\cdot \M$ is determined, up to a sign, by its
  second moment.
    \item If $R > 2M$, then every vector $x \in A \cdot \M$ is determined, up to a sign, by its second moment.
\end{enumerate}
\end{theorem}

\begin{theorem} \label{thm.rotation_multfree}
    Let $V$ be an $N$-dimensional representation of a compact group $G$,
  which decomposes as a sum of $R$ distinct irreducible representations
   $V  = V_1 \oplus \ldots \oplus V_R$.
  Let $\mathcal{M}$ be a semi-algebraic subset of $\RR^N$ of dimension $M$. Then,  for a generic rotation $A \in \SO(N)$, the following hold: 
\begin{enumerate}
\item If $R > M+1$,
  then a generic vector $x \in A\cdot \M$, $x$ is determined, up to a sign, by its
  second moment
\item If  $R > 2M+1$, then every vector $x \in A \cdot \M$ is determined, up to a sign, by its second moment.
\end{enumerate}
\end{theorem}

\begin{corollary}[Sample complexity]
Suppose that $n,\sigma\to\infty$.    Under the conditions of Theorem~\ref{thm.main_multfree} and Theorem~\ref{thm.rotation_multfree}, the signal $x$ can be determined uniquely, up to a sign, from only $n /\sigma^4\to\infty$ observations. 
\end{corollary}

Note that the results depend only on the number of irreducible representations $R$, and not on the dimension $N$. In the special case of phase retrieval,
$R=\lfloor{\frac{N}{2}} \rfloor +1$.
In Section~\ref{sec:cryoEM} we discuss the extension to general representations. 

\subsection{Proof of the main results}

The proofs of Theorem~\ref{thm.main_multfree} 
follows the same strategy as the proof
of Theorem~\ref{thm.main} and we will
only sketch it. Likewise, the proof of Theorem~\ref{thm.rotation_multfree}
follows the same strategy 
as the proof of Theorem~\ref{thm.rotation} and reduces to proving
Lemma~\ref{lem.Fxy_multfree}, which generalizes
Lemma~\ref{lem.dimFxy}. We do this in Section~\ref{sec.cruxproof}.

We begin with a general formulation that allows us to go between the formulations
of the second moment given in~\cite{bendory2022sample, edidin2023generic} and the notion
of separating invariants given in~\cite{dym2022low}.

If we view $\R^N$ as an orthogonal representation of a compact group
  $G$ which decomposes as a sum of $R$ distinct irreducible representations of dimensions
  $N_1, \ldots , N_R$, then the {\em second moment}  can be expressed
  (for a suitable choice of coordinates) as
  the function
  $m^2_{N_1, \ldots , N_R}(\cdot ) \colon \R^N \to \R^R$ given by
  \begin{equation} \label{eq.secondmoment}
    (x_1, \ldots, x_N)  \mapsto
    \left( x_1^2 + \ldots + x_{N_1}^2, x_{N_1+1}^2 + \ldots+ x_{N_1 + N_2}^2, \ldots , x_{N_1 +\ldots N_{R-1} +1}^2 + \ldots  + x_{N}^2\right).
\end{equation}
  This was proven in \cite{bendory2022sample} for complex representations
  and formula ~\eqref{eq.secondmoment} follow as long
  as the representations $V_\ell$ remain irreducible when the field of scalars
  is extended to $\CC$. (This occurs, for example, for all irreducible representations of $\SO(3)$.)
  In Appendix \ref{app.irreps} we explain how to get the corresponding
  statement for sums of arbitrary real irreducible representations.

The decomposition in \eqref{eq.secondmoment} implies that if $A$ is any $N \times N$ matrix with rows $(w_1, \ldots , w_N)$, 
  then
\begin{equation} \label{eq.momentsample}
P_{N_1, \ldots ,N_R}(x;A) = m^2_{N_1, \ldots , N_R}(Ax),
\end{equation}  
where 
$P_{N_1, \ldots , N_R}(\cdot;A) \colon \R^N \to \R$ is defined by the formula 
\begin{equation} \label{eq.Pfunction}
x \mapsto \left(p_{N_1}(x,\w_1), p_{N_2}(x,\w_2), \ldots , 
  p_{N_r}(x,\w_R)\right),
  \end{equation}
and $p_{N_i} \colon \R^N \times \R^{N_i \times N} \to \R$ is
the function defined by the formula
\begin{equation*}
p_{N_i}(x,\mathbf{v}) = \sum_{j=1}^{N_i} \langle x, v_j \rangle^2.\end{equation*}

(Here, $\mathbf{v} \in \R^{N_i \times N}$ corresponds to an $N_i$-tuple
$(v_1, \ldots , v_{N_i})$ of vectors in $\R^N$.)
This once again implies that if $A$ is an arbitrary $N \times N$
matrix, the second moment measurements of $Ax$ are separable with respect
to the rows of $A$.

If $\A \subset \R^N \times \R^N$ is a semi-algebraic set, let  
\begin{equation*}
\A(x,y)=\{A\in \A\,|\, P_{N_1, \ldots , N_R}(x;A)=P_{N_1, \ldots , N_R}(y;A) \}. \end{equation*}

The following generalization
of Lemma \ref{lem.AA} holds for the general second moment.
\begin{lemma} \label{lem.AAmultfree}
  Let $\mathcal{M}$ be a semi-algebraic subset of $\RR^N$ of dimension $M$, and let $\A \subseteq \RR^{N\times N}$ be a semi-algebraic set such that
  \begin{equation*}
 \dim(\A(x,y))\leq \dim(\A)-K, \quad  \forall x,y\in \M, \text{ with } x\not \sim y.\end{equation*}

  Then, for a generic invertible matrix $A \in \A$, the following hold:
\begin{enumerate}
\item If $K > M$,
  then a generic vector $x \in A\cdot \M$  is determined up to sign from its
  second moment.
\item If $K > 2M$, then every vector $x \in A \cdot \M$ is determined up to
  sign form its second moment.
\end{enumerate}
\end{lemma}

By Lemma \ref{lem.AAmultfree} it suffices to uniformly bound
the dimensions of the real algebraic sets
\begin{equation*}
\R^{N \times N}(x,y) = \{A \in \R^{N \times N}| P_{N_1, \ldots , N_R}(x;A) = P_{N_1, \ldots , N_R}(y;A)\},\end{equation*}
 and
\begin{equation*}
\SO(N)(x,y) = \{A \in \SO(N)|  P_{N_1, \ldots , N_R}(x;A) =
  P_{N_1, \ldots , N_R}(y;A)\},\end{equation*}

over all pairs $x \not \sim y$ of vectors in $\R^N$,
where the function $P_{N_1, \ldots , N_R}(\cdot;A)$ was defined
in \eqref{eq.Pfunction}.

Theorem \ref{thm.main_multfree} then follows from the following extension
of Lemma \ref{lem.RNN}.
\begin{lemma}\label{lem.RNN_multfree}
Let $N$ be a natural number. For every $x,y\in \RR^N$ such that $x \not \sim y$, consider the semi-algebraic set 
\begin{equation*}
\RR^{N\times N}(x,y)=\{A\in \RR^{N \times N}| \quad P_{N_1, \ldots , N_R}
(x\;A)= P_{N_1, \ldots , N_R}(y\;A) \} .\end{equation*}

Then, $\dim(\RR^{N\times N}(x,y)) \leq N^2-R$.
\end{lemma}
\begin{proof}
  The real algebraic set $\RR^{N \times N}(x,y)$ decomposes as a product
  $W^{N_1}(x,y) \times \ldots W^{N_R}(x,y),$ where
  \begin{equation*}
W^{N_k}(x,y) = \{(w_1, \ldots
  w_{N_k})\in \R^{N \times N_k}| p_{N_r}(x,w_1, \ldots , w_{N_k}) = p_{N_k}(x,w_1, \ldots , w_{N_r})\}.\end{equation*}

  The sets $W^{N_k}(x,y)$ all have dimension at most $NN_k -1$
  so we conclude that
  \begin{equation*}
\dim( \RR^{N \times N}(x,y)) \leq N^2 -R.\end{equation*}

\end{proof}

Theorem \ref{thm.rotation_multfree} follows from the following extension of Proposition \ref{prop.crux}.
 \begin{proposition}\label{prop.crux_multfree}
  For any $x, y \in \R^N$ with $x \not \sim y$, 
  \begin{equation*}
\dim \SO(x,y) \leq \dim \SO(N) -R+1.\end{equation*}
 .
\end{proposition}
 \begin{proof}
  
Once again, we realize $\SO(N) =F_N \to F_{N-1} \to \ldots F_1$ as a tower
   of sphere bundles.  Let $S_k = \sum_{i=1}^k N_k$ denote the partial sums of $N_i$ and 
consider the intermediate tower
   \begin{equation*}
F_{S_R} = \SO(N)  \to F_{S_{R-1}} \to \ldots F_{S_1} = F_{N_1}.\end{equation*}

For each $F_{S_j}$ in this tower, we define the real algebraic subset
\begin{equation*}
F_{S_j}(x,y) = \{(\w, \tilde{\w}) \in F_{S_j}| \w \in F_{S_{j-1}}(x,y), \
\tilde{\w} = (w_{S_{j-1}+1},\dotsm w_{S_{j}}),\ p_{N_j}(x, \tilde{\w}) = p_{N_j}(y,\tilde{\w}) \}.\end{equation*}

Proposition  \ref{prop.crux_multfree} then follows from the following generalization
of Lemma \ref{lem.dimFxy} whose proof is given in the next section. \end{proof}
 \begin{lemma} \label{lem.Fxy_multfree}
   Given $F_{S_k}(x,y)$ as above, we have $\dim F_{S_k}(x,y) \leq \dim F_{S_k} - k$ for $k=1,\dots, R-1$.
 \end{lemma}
 \subsection{Proof of Lemma \ref{lem.Fxy_multfree}} \label{sec.cruxproof}
 We begin with a preliminary lemma.
\begin{lemma} \label{lem.firststep} Let $x, y$ be fixed vectors in $\RR^N$ and let $L$ be an $r$-dimensional linear subspace of $\RR^N$. 
	For $1 \leq d < r$, let $M_d$
	be the $\sum_{i=1}^d (r-i)$-dimensional manifold consisting
	of $d$-tuples unit vectors $(w_1,\dots, w_d)$ with
        $\langle w_i, w_j \rangle =0$ for all $i\neq j$. If $x \pm y \notin L^\perp$, then
		$\dim M_d(x,y) \leq \dim M_d - 1$, where \begin{equation*}
M_d(x,y) = \{(w_1,\dots, w_d) \in M\,|\, p_d(x,w_1,\dots, w_d) = p_d(y,w_1, \dots, w_d)\},\end{equation*}

and $p_d(x,w_1, \ldots, w_d) = \sum_{i=1}^d \langle x,w_i \rangle^2$.
\end{lemma}

\begin{remark} \label{rem.stiefel}
  The manifold $M_d$ consisting of $d$-tuples of orthonormal
  vectors in an $r$-dimensional vector space $L$ is the set of real points
  of the Steifel manifold $S(d,r) = \{A \in \mathbb{C}^{d \times r} \,|\, AA^T = I_d\}$. The
  complex variety
  $S(d,r)$ is irreducible because it consists of a single orbit under the action of the connected (hence irreducible) algebraic group $\SO(r)$ \cite[Lemma 2.4]{brysiewicz2021degree}.
\end{remark}
\begin{remark}\label{mxyrmk}
  The condition that $r > d$ is necessary for the following reason.  
  If $r= d$, then a $d$-tuple of orthonormal vectors in $L$ is an orthononormal
  basis for $L$, and any two orthonormal bases are related by a rotation.
  On the other hand, if $A$ is a $d \times d$ rotation matrix, then
  $p_d(x,w_1, \ldots, w_d) = p_d(x, Aw_1, \ldots, Aw_d)$.
Hence, if
$\dim L = d$, then $M_d(x,y) =
\emptyset$ or $M_d(x,y) = M_d$.
\end{remark}

\begin{proof}
We begin with the case that $d =1$ because its proof is  simple.
  The equality $p_1(x,w) = p_1(y,w)$ simplifies to yield
  \begin{equation*}
 \langle x + y, w \rangle  \langle x-y, w \rangle = 0.\end{equation*}

  Since $x\pm y \notin L^\perp$, the set of all of $w\in M_1 = S^{r-1}$ satisfying
  $\langle x+y, w \rangle=0$ or $\langle x-y, w \rangle = 0$ is the intersection
  of $S^{r-1}$ with the union
  of two hyperplanes. 
  \medskip

  Now we consider the case $d > 1$.
  Since $x\pm y \not\in L^\perp$, we can not have $x=0 \bmod L^\perp$ and $y=0 \bmod L^\perp$. Without loss of generality, assume $x\neq 0 \bmod L^\perp$. Since $M(x,y) = M(\lambda x, \lambda y)$ for any non-zero scalar $\lambda$, we can further normalize and assume $x$ is a unit vector. We now choose an orthonormal basis
	$e_1,\dots, e_r$ for $L$ such that we set $x=e_1 \bmod L^{\perp}$ and
	choose $e_2$ to satisfy $y = ae_1 + be_2 \bmod L^{\perp}$ for some $a,b
	\in \RR$ (note that if $y\in L^\perp$, the choice of $e_2$ is
	arbitrary and we simply set $a=b=0$). In coordinates given by this
	basis, the manifold $M_d$ is a real algebraic subset of $
 \RR^{dr}$
defined by the following equations:
\begin{equation}
\alpha_{i1}\alpha_{j1} + \ldots + \alpha_{ir} \alpha_{jr}  = \delta_{i,j}
\end{equation} 
	where $w_i = [\alpha_{i1},\dots, \alpha_{ir}]^T$, $1 \leq i \leq j
        \leq d$, and $\delta_{i,j}$ denotes the Kronecker delta function.
The subset $M_d(x,y) \subset M_d$
is defined by an additional constraint
	\begin{equation} \label{eq.key} 
		\sum_{i=1}^{d} \alpha_{i1}^2 - (a\alpha_{i1}+b\alpha_{i2})^2 = 0.
	 \end{equation}
Note that the real algebraic condition $\eqref{eq.key}$ vanishes
identically when $b=0$ and $a=\pm 1$, which is precisely the condition that
either $x+y \in L^{\perp}$ (when $a=-1$) or $x-y \in L^\perp$ (when
$a=1$). Since $M_d$ is an irreducible variety, to obtain the desired dimension bound it suffices to show that
equation \eqref{eq.key} cannot vanish identically on all of $M_d$
unless $x \pm y \in L^\perp$.

Suppose $(a,b)\neq (\pm 1,0)$.  If $b \neq 0$, let $w_i =
e_{i+1}$ for $i=1, \ldots d$. (Note that we require
$d < r$, so $d+1 \leq r$.) Substituting this into the left-hand
side of equation \eqref{eq.key}, we get $b^2 \neq 0$. Thus $(w_1,
\ldots, w_d) \in M_d \setminus M_d(x,y)$. 
If $b =0$ then we know that $a \neq \pm 1$, so $|a| \neq 1$.  Now, take
$w_i = e_i$ for $i = 1, \ldots , d$. Substituting this into the left
hand-side of equation $\eqref{eq.key}$, we get $1-a^2 \neq 0$.  Hence,
$(w_1, \ldots , w_d) \in M_d \setminus M_{d}(x,y)$ in this case as
well. Therefore, $M_d(x,y)$ is a
proper algebraic subset of $M_d$.
\end{proof}
\begin{proof}[Proof of Lemma \ref{lem.Fxy_multfree}]
 To prove Lemma \ref{lem.Fxy_multfree} we will use induction and
  Lemma \ref{lem.firststep} to compute the dimension of the fibers
  of the map $F_{S_{k}}(x,y) \to F_{S_{k-1}}(x,y)$ for $k \geq 1$. 
  However, in order to accurately bound the dimension of $F_{S_k}(x,y)$ we will
  need to introduce an auxiliary
  real algebraic subset
\begin{equation*}
L_{S_k}(x,y) = \{\mathbf{w} \in F_{S_k}(x,y)| x\pm y \in \Span\w\}.\end{equation*}
   
The reason we consider the subspace $L_{S_{k}}(x,y)$ is that the fibers of
the map $F_{S_k}(x,y) \to F_{S_{k-1}}(x,y)$ have larger dimension over points
of $L_{S_{k-1}}(x,y)$ and so we will also need to estimate the dimension
of $L_{S_{k-1}}(x,y)$.
In particular, we will use induction to simultaneously prove two dimension bounds:
\begin{equation} \label{eq.dimFk}
  \dim F_{S_{k}}(x,y) \leq \dim F_{S_{k}}-k,
  \end{equation}
and
\begin{equation} \label{eq.dimLk}
  \dim L_{S_{k}}(x, y) \leq \dim F_{S_{k}} - (k-1)- (N-S_k). \end{equation}

Once we establish the base cases of \eqref{eq.dimFk} and \eqref{eq.dimLk} corresponding
to $k=1$, the induction will proceed as follows:
\begin{enumerate}
\item Assume by induction that the bounds hold for
  $\dim F_{S_{k-1}}(x,y)$ and $\dim L_{S_{k-1}}(x,y)$ for $k\geq 1$, and show
\begin{equation*}
\dim F_{S_k}(x,y) \leq
\max\{\dim F_{S_{k}} - k, \dim F_{S_k} - k - (N-S_{k-1}-2)\}\end{equation*}

and use the induction hypothesis on $\dim L_{S_{k-1}}(x,y)$
to conclude that the dimension bound holds for $F_{S_{k}}(x,y)$. 
\item
Get the dimension bound for $L_{S_k}(x,y)$ using the bounds for
$\dim L_{S_{k-1}}(x,y)$ and $\dim F_{S_{k-1}}(x,y)$.  \end{enumerate}

\medskip

\noindent{\bf Basis step.}
By definition $F_{S_1}(x,y) = \{\mathbf{w} \in F_{S_1}\,|\, p_{N_1}(x,\w) = p_{N_1}(y,\w)\}$. Applying Lemma \ref{lem.firststep} with $L = \R^N$ we conclude
that $\dim F_{S_1}(x,y) \leq \dim F_{S_1} - 1$.
If $\mathbf{w} = (w_1, \ldots , w_{N_1}) \in L_{S_1}(x,y)$ then
by definition, either $x +y$ or $x-y$ is in $\Span \w$, which is an $N_d$
dimensional subspace of $\R^N$. The locus of $N_d$-dimensional linear subspaces of $\R^N$ containing a fixed vector is defined by the Schubert condition 
$(N-N_1, 0, \ldots ,0)$ and therefore has codimension $N-N_1$ in
the Grassmannian $\Gr(N_1, N)$ of $N_1$-dimensional linear subspaces of $\R^N$
\cite[Chapter 5]{griffiths1978principles}.
Since $F_{N_1}$ is the orthonormal frame bundle over the Grassmannian
$\Gr(N_1, N)$, the set of $\w \in F_{N_1}$ that contain either $x+y$ or $x-y$
has dimension $\dim F_{N_1} - (N-N_1)$. Since
$L_{S_1}(x,y)$ is contained in this subset, we conclude that
$\dim L_{S_1}(x,y) \leq \dim F_{S_1}- (N-S_1)$, since $S_1= N_1$.

\medskip

\noindent{\bf Induction step.}
Assume by induction that $\dim F_{S_{k-1}}(x,y) \leq \dim F_{S_{k-1}} -(k-1)$ and that
$\dim L_{S_{k-1}}(x,y) \leq \dim F_{S_{k-1}} -(k-2) - (N-S_{k-1})$
By definition,
\begin{align*}F_{S_k}(x,y) &=\{ (\mathbf{w},\tilde{\w})| \quad \mathbf{\w}\in
  F_{S_{k-1}}(x,y) \text{ and } p_{N_1}(x, \tilde{\w})=
p_{N_1}(y, \tilde{\w}) \\ 
 & \text{ and } \tilde{\w} = (w_{S_{k-1} + 1},\ldots,w_{S_{k}}) \text{ are pairwise orthonormal in} \Span(\w)^\perp\}, 
 \end{align*}
Let $\tau_{S_{k}}: F_{S_k} \to F_{S_{k-1}}$ be given by projection onto the first factor.
The fibers of $\tau_{S_k}$ have dimension
\begin{equation*}
\Gamma_k = N-(S_{k-1}+1) + \dots + N-(S_{k-1}+N_k) = N_k(N-S_k) + \binom{N_k+1}{2}.\end{equation*}

We wish to bound the dimension of the fibers of the restricted
map $\tau_{S_k} \colon F_{S_k}(x,y) \to F_{S_{k-1}}(x,y)$.
Consider the stratification of $F_{S_k}(x,y)$ given by
\begin{equation*}
F_{S_k}^{(1)}(x,y)= \{(\w, \tilde{\w})\in F_{S_k}(x,y)
| \w \in F_{S_{k-1}} \setminus L_{S_{k-1}}(x,y)\},\end{equation*}
  
and
\begin{equation*}
F_{S_k}^{(2)}(x,y)= \{(\w, \tilde{\w})\in F_{S_k}(x,y)
| \w \in L_{S_{k-1}}(x,y)\},\end{equation*}
  
and let us analyze the dimensions of the fibers of $\tau_{S_k}$
over each of the two strata.
\begin{enumerate}

\item If $\mathbf{w} \in F_{S_{k-1}}(x,y)\setminus 
L_{S_{k-1}}(x,y)$ the
 fiber of $F_{S_{k}}(x,y) \to F_{S_{k}}(x,y)$ over $\mathbf{w}$ is the set
 of $N_{k}$-tuples of orthonormal vectors $\tilde{\w} \in \Span\mathbf{w}^\perp$
 that satisfy $p_{N_k}(x,\tilde{\w})
= p_{N_k}(y,\tilde{\w})$.
 Since we know that
 $\mathbf{w}$ lies in the complement of $L_{S_{k-1}}(x,y)$, $x\pm y \notin
 \Span(\mathbf{w})$. Hence, Lemma
 \ref{lem.firststep} applied with $d=N_k$ and $L = \Span(\mathbf{w})^\perp$
 implies that the dimension of
 the fibers over this locus is  $\Gamma_k -1$. Using this and \eqref{eq:conservation2}
we obtain
 \begin{equation*} \begin{split} \dim
 F_{S_{k}}^{(1)}(x,y) &\leq \dim
 F_{S_{k-1}}(x,y)+ \Gamma_k -1 \\ &\leq  \dim F_{S_{k}} - k.
\end{split}
\end{equation*}
\item If $\mathbf{w} \in L_{S_{k-1}}(x,y)$ then by definition of 
we know that either $x+y$ or $x-y$ belongs to
$\Span(\mathbf{w})$. This means that for any choice of $\tilde{\w}$ the equation
$p_{N_k}(x,\tilde{\w}) = p_{N_k}(y,\tilde{\w})$
 is automatically satisfied. Since no additional conditions are imposed
 on $(\mathbf{w}, \tilde{\w})$ the fibers of $F_{S_{k}}(x,y)
 \to F_{S_{k-1}}(x,y)$ have maximal possible dimension which is
 $\Gamma_k$. Using this and \eqref{eq:conservation2} we obtain
\begin{equation*}
 \begin{split} \dim F_{S_k}^{(2)}(x,y) & \leq \dim L_{S_{k-1}}(x,y)
 + \Gamma_k\\ & \leq \dim F_{S_{k-1}} - (k-2) - (N-S_{k-1}) + \Gamma_k\\
&= \dim F_{S_k} - k - (N-S_{k-1}-2).
 \end{split}
 \end{equation*}
 \end{enumerate}
 Compiling the dimension bounds from both strata we conclude when $N\geq S_{k-1}+2$, \begin{equation*}
\dim F_{S_k}(x,y) \leq \max\{\dim F_{S_{k}} - k, \dim F_{S_k} - k - (N-S_{k-1}-2)\} = \dim F_{S_{k}} - k.\end{equation*}

(Note that the requirement $N\geq S_{k-1}+2$ will always be satisfied if $k\leq R-1$.)

\medskip

We now compute $\dim L_{S_{k}}(x,y)$ by analyzing the dimension
of the fibers restricted projection map $\tau_{S_k} \colon L_{S_k}(x,y) \to F_{S_{k-1}}(x,y)$.
  \begin{enumerate}
  \item If $\w\in L_{S_{k-1}}(x,y)$, then certainly any pair
    $(\w,\tilde{\w})$ in $F_{S_{k}}(x,y)$ is in $L_{S_{k}}(x,y)$. In this case we get
	\begin{equation*}
		\begin{split}
			\dim \tau_{S_k}^{-1}(L_{S_{k-1}}(x,y)) &\leq \dim L_{S_{k-1}}(x,y) + \Gamma_k \\
			& \leq \dim F_{S_k} - k - (N-S_{k-1}-2).
		\end{split}
	\end{equation*}
      \item If $w \not\in L_{S_{k-1}}(x,y)$, then let $a = x+y -
        \Proj_{\Span(w)}(x+y)$ and $b=x-y-\Proj_{\Span(\w)}(x-y)$.
        By
        hypothesis, $a,b$ are non-zero vectors in $w^\perp$. An
        $N_k-$tuple $\tilde{\w}$ is an orthonormal basis for a $N_k$-plane in $\Span(\w)^\perp \cong \R^{N-S_{k-1}}$. The locus of such $N_k$-planes is the Grassmann manifold $Gr(N_k,N-S_{k-1})$, which
        has dimension $N_k(N-S_k)$. The subvariety of $N_k$-planes in
        $\Span(\w)^\perp$ containing a fixed vector is given by the Schubert
        condition $(N-S_k,0,\dots, 0)$ and therefore has codimension
        $N-S_k$. Hence, 
	\begin{equation*}
		\begin{split}
\dim \tau_{S_k}^{-1}(F_{S_{k-1}} \setminus L_{S_{k-1}}) &\leq \dim F_{S_{k-1}}(x,y) + \Gamma_k - (N-S_k)\\
				&\leq \dim F_{S_{k-1}} - (k-1) + \Gamma_k - (N-S_k)\\
				&\leq \dim F_{S_k} - (k-1) - (N-S_k).
 		\end{split}
	\end{equation*}
	\end{enumerate}
  Compiling the bounds from both cases, we conclude that
  \begin{equation*}
      \begin{split}
  \dim L_{S_k}(x,y) &\leq \max\{\dim F_{S_k} - k - (N-S_{k-1}-2),\dim F_{S_k} - (k-1) - (N-S_k)\} \\ &= \dim F_{S_k} - (k-1) - (N-S_k).
  \end{split}
  \end{equation*}
as desired.
\end{proof} 

\subsection{Example: Band-limited functions on the sphere $S^2$}
As an example,  we now discuss a special case of the MRA model~\eqref{eq:mra1}, where $V$ is the space of real band-limited functions on the sphere $S^2$ under the action of $\SO(3)$. 
This problem was studied, without considering prior information, in~\cite{fan2021maximum,bandeira2023estimation,liu2021algorithms}.
We give bounds on
the dimension of a semi-algebraic set whose generic translation can distinguish
real band-limited functions from their second moment.

Let $W_L$ the space of $L$-band-limited functions in $L^2(S^2)$. The space
$W_L$ decomposes
as a sum $W_L = V_0 \oplus \ldots \oplus V_L$, where each $V_\ell$ is an irreducible
representation of $\SO(3)$ of dimension $2\ell + 1$ with a basis of
degree $\ell$ spherical harmonic polynomials. Let $N = \sum_{\ell =0}^L (2\ell + 1) = (L+1)^2$ be the dimension of $W_L$. Then we have the following result.

\begin{corollary} \label{cor.so3}
 Let $\mathcal{M}$ be a semi-algebraic subset of $W_L$ of dimension $M$.
  \begin{enumerate}
\item    If $L+1>M$, then for a generic $A \in \GL(N)$,  the generic vector
  $x \in A \cdot \M$, is determined up to a sign by its second moment.
\item  If $L>M$, then for a generic $A \in \SO(N)$ the generic vector 
  $x \in A \cdot \M$ is determined up to a sign by its second moment.
  \end{enumerate}
  Under the conditions above, an $L$ band-limited function on the sphere, acted upon by random elements of $\SO(3)$, can be determined by $n/\sigma^4\to\infty$ observations when $n,\sigma\to\infty.$
\end{corollary}

\section{Future work} \label{sec:future_work}

\subsection{The sample complexity of cryo-EM with semi-algebraic priors} 
\label{sec:cryoEM}
Single-particle cryo-electron microscopy (cryo-EM) is a leading technology in constituting the structure of molecules~\cite{bendory2020single}. Under some simplified assumptions, the reconstruction problem in cryo-EM can be formulated as  estimating a signal $x$ (the 3-D molecular structure) from $n$ observations of the form 
\begin{equation} \label{eq:cryoEM}
 y_i = T(g_i\cdot x)+\varepsilon_i, \quad \varepsilon_i\sim\N(0,\sigma^2I), \quad i=1,\ldots,n,   
\end{equation}
where $g_i$ are random elements of the group of three-dimensional rotations $\SO(3)$ and 
$T$ is a tomographic projection acting by $Tx(u_1,u_2)=\int_{u_3}x(u_1,u_2,u_3)du_3.$ It is common to assume that $x$ can be represented by a collection of $R$-band-limited functions in $L^2(S^2)$ (as was considered in Corollary~\ref{cor.so3}). While $T$ is not a group action (for example, it is not invertible), it was shown that it does not affect the second moment of the observations, and thus the model~\eqref{eq:cryoEM} can be studied as a case of the MRA model~\eqref{eq:mra1}~\cite{kam1980reconstruction,bendory2022sample}. (The tomographic projection does affect higher-order moments.)

Our goal in future work is to extend Theorem~\ref{thm.main_multfree} and Theorem~\ref{thm.rotation_multfree} to any finite-dimensional representation of the MRA model~\eqref{eq:mra1}. In particular, we wish to derive conditions under which a molecular structure can be recovered from the second moment of the cryo-EM model~\eqref{eq:cryoEM}, where each irreducible representation has a multiplicity greater than one. This in turn may provide the theoretical background to recent results in the field, which are based on deep generative models, see for example~\cite{zhong2021cryodrgn,chen2021deep,kimanius2023data}.

\subsection{Analytical activation function}
In this work, we considered neural networks with semi-algebraic activation functions. We intend to extend these results for activation functions, which are not semi-algebraic but are \emph{analytic}, such as sigmoid, tanh, or swish activations. Such a result can be obtained by combining the techniques we used here with the ideas used in \cite{amir2023neural} to expand similar ideas from the algebraic to the analytic setting.

\section*{Acknowledgment}
The work of the first and third authors were supported by BSF grant 2020159. The work of the first author was also partially supported by NSF-BSF grant 2019752 and ISF grant 1924/21. The work of the second author was supported by ISF grant 272/23. The work of the third author was also supported by NSF-DMS 2205626.

\bibliographystyle{plain}

\appendix

\section{Background on Real Algebraic Geometry} \label{app.algebraic}

We include background material on real algebraic geometry. Much of this discussion appears in \cite{edidin2019geometry}, but we include it here to make the current paper self-contained.
\subsection{Real algebraic sets}
  A real algebraic set is a subset $X= V(f_1, \ldots ,f_r) \subset \R^M$
  defined by the simultaneous vanishing of polynomial equations
  $f_1, \ldots , f_r \in \R[x_1,\ldots, x_M]$.
  Note that any real algebraic set is defined by the single
  polynomial $F = f_1^2 +\ldots + f_r^2$.
Given an algebraic set $X = V(f_1, \ldots , f_m)$, we define the
Zariski topology on $X$ by declaring closed sets to be the
intersections of $X$ with other algebraic subsets of
$\R^M$. An algebraic set is {\em irreducible} if it is not the union
of proper algebraic subsets. An irreducible algebraic set is
called a {\em real algebraic variety}. Every algebraic set has a decomposition
into a finite union of irreducible algebraic subsets \cite[Theorem 2.8.3]{bochnak1998real}.

An algebraic subset of $ X \subset \R^M$ is irreducible if and only if
the ideal $I(X) \subset \R[x_1, \ldots , x_M]$ of polynomials
vanishing on $X$ is prime.  More generally, we declare an arbitrary
subset of $X \subset \R^M$ to be irreducible if its closure in the
Zariski topology is irreducible. This is equivalent to the statement
that $I(X)$ is a prime ideal \cite[Theorem 2.8.3]{bochnak1998real}.

Note that in real algebraic geometry, irreducible algebraic sets need not be
connected in the classical topology. For example, the real variety defined by the equation $y^2 - x^3 +x$ consists of two connected components.

\subsection{Semi-algebraic sets and their maps}
In real algebraic geometry, it is also natural to consider subsets
of $\R^M$ defined by inequalities of polynomials. A {\em semi-algebraic} subset of $\R^M$ is a finite union of subsets of the form:
\[ \{x \in \R^M; P(x) = 0\; \text{ and }\; \left(Q_1(x)>0, \ldots , Q_{\ell}(x) > 0\right)\}.\]
Note that if $f \in R[x_1, \ldots , x_M]$ the set $f(x) \geq 0$
is semi-algebraic since it is the union of the set $f(x) =0$ with the set $f(x) > 0$.

The reason for considering semi-algebraic sets is that the image of an
algebraic set under a real algebraic map need not be real
algebraic. For a simple example, consider the algebraic map $\R \to
\R,
x \mapsto x^2$. This map is algebraic, but its image is the
semi-algebraic set $\{x\geq 0\} \subset \R$. A basic result in real
algebraic geometry states that the image of a semi-algebraic set under
a polynomial map is semi-algebraic, \cite[Proposition 2.2.7]{bochnak1998real}.

A map $f \colon \X \subset \R^N \to \Y \subset \R^M$ of semi-algebraic
sets is {\em semi-algebraic} if the graph $\Gamma_f = \{(x,f(x))\}$ is
a semi-algebraic subset of $\R^N \times \R^M$.  For example, the map
$\R_{\geq 0} \to \R_{\geq 0}, x \mapsto \sqrt{x}$ is semi-algebraic
since the graph $\{(x,\sqrt{x})| x\geq 0\}$ is the semi-algebraic
subset of $\R^2$ defined by the equation $x=y^2$ and inequality $x
\geq 0$.  Again, the image of a semi-algebraic set under a
semi-algebraic map is semi-algebraic.

\subsection{Dimension of semi-algebraic sets}
A result in real algebraic geometry
\cite[Theorem 2.3.6]{bochnak1998real} states that any real semi-algebraic
subset of $\R^n$ admits a semi-algebraic homeomorphism\footnote{A semi-algebraic
  map $f \colon A \to B$ is a semi-algebraic homeomorphism if $f$ is bijective
  and $f^{-1}$ is semi-algebraic.}
to a finite disjoint union of 
hypercubes. Thus, we
can define the dimension of a semi-algebraic set $\X$
to be the maximal dimension of a hypercube
in this decomposition. 
If $\X$ is a semi-algebraic set, then the dimension of $\X$ defined in the previous paragraph can be shown to be equal to
the Krull dimension of the Zariski closure of $\X$ in $\R^M$ \cite[Corollary 2.8.9]{bochnak1998real}. This fact has several consequences:
\begin{enumerate}
\item If $\X$ decomposes into a finite union
$\{\X_\ell\}$ into semi-algebraic subsets, then $\dim \X = \max_\ell \dim \X_\ell$.
\item If $\Y$ is a semi-algebraic
subset of an algebraic set $X$ with $\dim \Y  < \dim \X$, then $\Y$ is a contained
in a proper algebraic subset of $\X$.
\item 
If $\X$ is an irreducible semi-algebraic set and $\Y$ is a proper algebraic subset
(i.e., $Y$ is defined by polynomial equations) then $\dim \Y < \dim \X$.
\end{enumerate}
We also need to understand the properties of the dimensions of semi-algebraic sets under semi-algebraic maps. To that end, we will use the following two properties:
\begin{enumerate}
\item Let $\X$ be a semi-algebraic set and $\pi \colon \X \to \R^p$ a semi-algebraic mapping. Then, $\dim \X \geq \dim \pi(\X)$ \cite[Theorem 2.8.8]{bochnak1998real}
\item  Let $\X \subset \R^n$ be a semi-algebraic set and let
  $\pi \colon \R^n \to \R^m$ be a polynomial map.
  Then according to \cite[Lemma
    1.8]{dym2022low},
  \begin{equation}\label{eq:conservation2}
  \dim \X \leq \dim \pi(\X) +\max_{y\in \pi(\X)}\dim(\pi^{-1}(y))\end{equation} 
\end{enumerate}  
\rev{\subsection{Grassmann manifolds}
The set $\Gr(M,V)$ of $M$-dimensional linear subspaces of an $N$-dimensional
vector space $V$ has the natural structure as an algebraic  manifold, called
the {\em Grassmannn manifold of $M$ planes in $V$}. The Grassmann manifold $\Gr(M,V)$
has dimension $M(N-M)$ and there are a number of ways to see the manifold
structure and compute the dimension.

Given an ordered basis $v_1, \ldots , v_M$
for an $M$-dimensional linear subspace
$\Lambda$, we can associate a full rank $M \times N$ matrix
$A_\Lambda = [v_1 \ldots v_M]$. Conversely, given a full rank $M \times N$
matrix~$A$, the columns of $A$ determine an ordered basis for an $M$-dimensional linear
subspace of $V$. Two matrices $A,A'$ correspond to the same linear subspace
if and only if there is an invertible $M \times M$ matrix $P$ such
that $A = PA'$. Hence, the Grassmann manifold can be identified as
the quotient $F(M,N)/\GL(M)$, where $F(M,N)$ is the set of full rank
$M \times N$ matrices. Since 
$F(M,N)$ has dimension $MN$ and $\GL(M)$
is a group of dimension $M^2$ which acts with trivial stabilizers, the quotient
has dimension $MN - M^2 = M(N-M)$.  The map $F(M,N) \to \Gr(M,V)$ 
maps a full-rank $M \times N$ matrix to the linear subspace of $\R^N$ spanned by the rows of the matrix. The fiber of the quotient map $F(M,N) \to \Gr(M,N)$ over a point of $\Gr(M,N)$ parametrizing
an $M$-dimensional linear subspace $\Lambda$ is the set of all possible bases for $\Lambda$ as a subspace of $\R^N$ or equivalently the set of $M \times N$ matrices whose range is $\Lambda$.
For more on Grassmann manifolds, see \cite[Chapter 1, Section 5]{griffiths1978principles}.
\subsection{The structure of the groups $\OO(N)$ and $\SO(N)$}
The group $\OO(N)$ of orthogonal matrices is a real algebraic subset of $\R^{N \times N}$ defined by the condition $A^T A = \Id$.
An element of $\OO(N)$ is an $N$-tuple of mutually orthogonal unit vectors $(w_1, \ldots , w_N)$ with $w_i \in \R^N$. If we denote by
$F_k$ to be the set of $k$-tuples $(w_1, \ldots , w_k)$ of mutually orthogonal unit vectors then $\OO(N) = F_N$
and we have a sequence of projection maps $F_N \stackrel{\pi_N} \to F_{N-1} \to  \ldots \stackrel{\pi_2} \to F_{1}$ where the map $\pi_{k+1} \colon F_{k+1} \to F_k$ is defined by dropping
the last vector; i.e., 
$(w_1, \ldots , w_{k+1}) \mapsto (w_1, \ldots , w_k)$. The space $F_1$ is the unit sphere $S^{N-1}$ and if
$(w_1, \ldots , w_k) \in F_k$,  then the fiber $\pi_k^{-1}(w_1, \ldots w_k)$ consists of all $k+1$ tuples $(w_1, \ldots , w_k, w)$
where $w$ is a unit vector orthogonal to $(w_1, \ldots , w_k)$. Thus, the fiber can be identified with the unit sphere in
the $(N-k)$-dimensional subspace $(w_1, \ldots, w_k)^\perp$. For this reason, we say that the sequence of maps $F_N \ldots F_{N-1} \to \ldots \to F_1$ is a {\em tower of sphere bundles} since the fibers of each map in the sequence are spheres.

The group $\SO(N)$ of orthogonal matrices of determinant one can likewise be realized as a tower of sphere bundles
$\SO(N) = F_N \to F_{N_1} \to \ldots \to F_1$. The only difference is that the fibers $F_N \to F_{N-1}$ consist of a single
unit vector. The reason is that if $(w_1, \ldots , w_{N-1})$ is an $(N-1)-$tuple of mutually orthogonal unit vectors, there is a unique
unit vector $w \in (w_1, \ldots , w_{N-1})^\perp$ such that the matrix $(w_1, \ldots , w_{N-1}, w)$ has determinant one.
}

\section{Second moments for multiplicity free representations }\label{app.irreps}
A real representation $V$ of a compact 
 group is  {\em multiplicity free}
  if $V$ decomposes as a sum of distinct (non-isomorphic) irreducible
real  representations
\begin{equation*}
V = V_1 \oplus \ldots \oplus V_R.\end{equation*}

If $V$ is multiplicity free then we can express
$f \in V$ as $f = (f[1], \ldots , f[R])$, where $f[\ell] \in V_\ell$.
The second moment of a signal $x\in  V$ is defined by the formula:
\begin{equation}\label{eq:second_moment}
 \mathbb{E}[ff^*] = \int_{g\in G} (g\cdot f) (g\cdot f)^*dg .  
\end{equation}
\begin{proposition}
  If
  $V = V_1 \oplus \ldots \oplus V_R$ is a multiplicity free
  representation of a compact group $G$
  then the {\em second moment} matrix can be expressed
as a composition of an injective linear function $L:\RR^R\to \RR^{N\times N} $ with
  the function
  $m^2_{N_1, \ldots , N_r}(\cdot ) \colon V \to \RR^r$ given by
  \begin{equation*}
    f = (f[1], \ldots, f[R])  \mapsto
    \left( |f[1]|^2, \ldots, |f[r]|^2\right),
\end{equation*}
  where $f[\ell]$ is the component of $f$ in the $\ell$-th irreducible
  summand $V_\ell$.
\end{proposition}
As mentioned in the main text, this proposition was proven in \cite{bendory2022sample} for complex representations.
\begin{proof}
  By definition, the second moment is a map $V \to \Hom_G(V,V)$,
  where $\Hom_G(V,V)$ denotes the vector space of $G$-invariant linear transformation $V \to V$. Since $V = \oplus V_\ell$, $\Hom_G(V,V)$ decomposes
  as a sum $\oplus_{k,\ell} \Hom_G(V_k, V_\ell)$. If $k \neq \ell$ then
  $\Hom_G(V_k, V_\ell) = 0$ because any map of finite dimensional irreducible
  representations of a group is either 0 or an isomorphism. (This is a special case of Schur's lemma and holds over both $\R$ and $\CC$.)
  It follows that, with respect to a suitable basis on $V$, the matrix representing
  $m^2_{N_1, \ldots , N_R}(x)$ is a block diagonal matrix, where the $\ell$-th block has size
  $N_\ell \times N_\ell$ where $N_\ell = \dim V_\ell$ and is given
  by $M_\ell =\int_G (g\cdot f[\ell])(g\cdot f[\ell])^*$.

  We claim that matrix $M_\ell$ is necessarily a scalar multiple of the diagonal.
  The reason is that $M_\ell$ is a symmetric matrix, so it necessarily has real
  eigenvalues. Since $M_\ell$ defines a $G$-invariant linear transform, 
  its eigenspaces are $G$-invariant subspaces. Hence, $M_\ell$ can have only
  a single eigenvalue, i.e., it is a scalar multiple of identity.
(Note that if the representation $V_\ell$ remains irreducible after extending the scalars to $\CC$ this follows directly from Schur's lemma.)
  Moreover, if $M_\ell = \lambda_\ell \Id_{N_\ell}$ we can determine the
  scalar $\lambda_\ell$ as in the proof of \cite[Proposition 2.1]{bendory2022sample}. Namely,
  \begin{equation*}
\trace\;M_\ell = N_\ell \lambda_\ell  \int_G \trace \left((g \cdot f[\ell]) (g \cdot f[\ell])^*\right) = |f[\ell]|^2,\end{equation*}

  where the latter equality holds because $g$ acts by orthogonal transformations.
\end{proof}

\end{document}